%% file: main.tex
\documentclass[11pt]{article}
\usepackage{geometry}
\usepackage{amsmath, amssymb, amsthm}
\usepackage{graphicx}
\usepackage{listings}
\usepackage{tikz}
\usepackage{float}
\usepackage{booktabs}
\usepackage{enumitem}
\usepackage{hyperref}
\usepackage{fancyhdr}
\usepackage{titlesec}
\usepackage{mathtools}
\usepackage{xcolor}
\usepackage{authblk}

\newtheorem{rulex}{Rule}
\newtheorem{theorem}{Theorem}

\newtheorem{corollary}{Corollary}

\usetikzlibrary{positioning,arrows,matrix}

\title{\textbf{Linguine: A Natural-Language Programming Language\\ with Formal Semantics and a Clean Compiler Pipeline}}
\author{Lifan Hu \\
\texttt{lifan.hnus@gmail.com}}
\affil{School of Computing, National University of Singapore}
\date{\today}

\begin{document}
\sloppy
\maketitle

\begin{abstract}
\input{sections/abstract}
\end{abstract}

\section{Introduction}
\input{sections/introduction}

\section{Related Work}
\input{sections/related_work}

\section{System Design}
\input{sections/system_design}

\section{Formal Semantics}
\input{sections/formal_semantics}

\section{Preliminary Evaluation}
\input{sections/evaluation}

\section{Discussion}
\input{sections/discussion}

\section{Conclusion}
\input{sections/conclusion}

\input{sections/appendix}

\bibliographystyle{plain}
\bibliography{references}

\end{document}

%% file: sections/abstract.tex
Linguine is a natural-language-inspired programming language that enables users to write programs in a fluent, controlled subset of English while preserving formal semantics. The language introduces anaphoric constructs—such as pronoun variables (e.g., \texttt{it}, \texttt{them})—that are statically resolved via a referent-tracking analysis layered atop a Hindley–Milner-style type system. Every pronoun is guaranteed to be both unambiguous and well-typed at compile time.

The Linguine compiler pipeline comprises lexing, parsing, clause graph construction, desugaring into a typed intermediate representation (IR), type inference, and abstract interpretation. This pipeline statically detects semantic errors such as undefined, misordered, or type-inconsistent references. A lightweight code generation stage currently targets Python.

This paper formalizes the core language, defines its typing and operational semantics, and proves the soundness of its pronoun resolution mechanism. An initial evaluation demonstrates that Linguine enables the expression of concise and human-readable programs while supporting early static error detection.

Linguine represents a step toward programming systems that prioritize human linguistic intuition while remaining grounded in formal methods and type-theoretic rigor.

%% file: sections/introduction.tex
\subsection{Motivation}

Programming languages that resemble natural language have long been proposed as a way to improve code readability and accessibility. Early systems such as \textsc{Cobol} demonstrated the potential of English-like syntax in business contexts, but also revealed fundamental challenges related to ambiguity and formal reasoning. More recent domain-specific languages, such as Inform 7, achieve highly naturalistic surface forms within constrained settings. However, general-purpose programming remains dominated by symbolic syntax designed for unambiguous parsing and rigorous static analysis.

Most modern compilers are built atop formal foundations: a context-free grammar $\mathcal{G}$, a typing relation $\vdash$, and a small-step or big-step operational semantics $\rightsquigarrow$. These frameworks offer precision but appear incompatible with the fluidity and referential structures of human language—particularly phenomena such as anaphora and ellipsis.

\subsection{The Linguine Hypothesis}

\textit{Linguine} is an experimental programming language designed to explore whether a natural-language-inspired surface syntax can coexist with a formally analyzable compiler pipeline. The source language resembles controlled English and supports a constrained set of sentence structures. Programs are parsed deterministically into an abstract syntax tree $A$, then lowered via desugaring and static single-assignment (SSA) transformation into a typed intermediate representation $I$, before being emitted as target code $T$ (currently Python):

\[
A \xrightarrow{\text{desugar}} A' \xrightarrow{\text{SSA}} I \xrightarrow{\text{codegen}} T
\]

A central feature of Linguine is its support for \textit{pronoun variables}—tokens such as \texttt{it}, \texttt{this}, or \texttt{them} that refer to previously defined entities. These are resolved statically using a referent memory $\rho : \mathcal{P} \to \mathsf{Ref}_\bot$, where $\mathcal{P}$ is the set of permitted pronouns. A pronoun is considered valid only if it refers to a defined referent with a valid type:

\[
\rho(p) = e \quad \text{and} \quad \Gamma \vdash e : \tau
\]

Static analysis ensures every pronoun is resolved deterministically and is well-typed. An abstract interpretation pass verifies that no pronoun remains undefined at runtime; ambiguity or unresolved references are reported as compile-time errors.

\subsection{Prototype Status and Roadmap}

The current Linguine implementation is a working prototype that validates the core language design. It includes a lexer, parser, referent-tracking system, Hindley–Milner-style type inference, abstract interpreter, and backend code generator. The compiler emits Python code, and preliminary support for LLVM IR is planned.

At present, the language supports basic constructs: assignments, conditionals, and arithmetic comparisons. Most test programs are under 15 lines and compile in milliseconds. Semantic errors—such as undefined pronouns or type mismatches—are caught early during compilation.

Future development will focus on expanding syntactic coverage, introducing user-defined functions and structured data types, and incorporating optimization passes such as constant propagation and dead code elimination.

\subsection{Contributions}

This paper presents the design and implementation of \textit{Linguine}, a prototype language that adopts a controlled natural-language syntax while maintaining formal rigor. The key innovation is a referent-aware pronoun system, enabling anaphoric constructs (e.g., \texttt{it}, \texttt{them}) that are resolved statically through a referent memory integrated with a Hindley–Milner type system. The compilation pipeline includes deterministic parsing, clause graph construction, desugaring, SSA transformation, type inference, abstract interpretation for disambiguation, and Python code generation. The system performs compile-time detection of semantic errors, including unresolved references and type mismatches.

A formal semantics has been initiated to ground the system’s core design. Preliminary tests demonstrate that Linguine compiles simple programs efficiently while preserving both human readability and static verifiability. While early-stage, the prototype suggests that naturalistic syntax and structured compilation are not incompatible—and may, in fact, mutually reinforce one another.

%% file: sections/related_work.tex
\label{sec:related_work}

Natural-language programming has been explored since the 1960s, yet each generation has grappled with the same fundamental tension: reconciling surface-level readability with formal precision and analyzability.

\paragraph{English-like general-purpose languages.}
\textsc{Cobol} demonstrated that English-like verbs and clauses could make business logic accessible to non-specialists, but its expansive grammar quickly exposed the limitations of unrestricted phrasing—namely, ambiguity and lack of formal rigor \cite{cobol60,sammet1969history}.
AppleScript continued this trajectory for desktop automation, enabling statements such as
\texttt{if the firstNumber is greater than the secondNumber then …}, but its implicit coercion rules often led to surprising behavior \cite{appleScriptGuide,hamish2010}.
Both cases highlight the need for a \emph{restricted} linguistic subset coupled with a formally defined semantics—a stance explicitly adopted by Linguine.

\paragraph{Controlled languages and domain-specific systems.}
Biermann’s \textsc{NLP} system showed that novice users could write algorithms in constrained English when guided by an interactive editor \cite{biermann1983nlp}.
Inform 7 extends this philosophy: its source code resembles naturalistic English prose but compiles into predicate logic over a domain-specific ontology of rooms, objects, and actors \cite{nelson2005}.
While Inform validates the controlled-language approach, its semantics are tightly coupled to a narrative domain.
Linguine generalizes this strategy to a more traditional imperative computational model.

\paragraph{Pronouns and anaphora.}
Context-sensitive identifiers appear sporadically in mainstream languages—e.g., Perl’s \texttt{\$\_}, Java’s \texttt{this}, and positional parameters in Unix shells—but are rarely integrated with static typing.
A Microsoft patent proposes user-definable programming-language pronouns \cite{todd2000}, but does not formalize their resolution.
To our knowledge, Linguine is the first system to treat English pronouns as first-class programming constructs with a provably unambiguous static semantics.

\paragraph{Type inference and semantics-preserving desugaring.}
Hindley–Milner type inference is widely used to provide implicit typing in functional languages \cite{milner1978theory}, and semantics-preserving desugaring is foundational in compilers from Scheme to Scala.
Linguine combines both: its surface syntax desugars into a small core language, after which Algorithm~W infers principal types.
Earlier natural-language-like systems such as NaturalJava used dynamic typing and could not guarantee static soundness \cite{price2000naturalJava}.

\paragraph{Abstract interpretation in front-end analysis.}
Cousot and Cousot’s framework for abstract interpretation underlies many compiler optimizations \cite{cousot1977abstract}, but its use in front-end semantic validation is comparatively rare.
Java’s definite-assignment check is a notable exception.
Linguine generalizes this idea: it defines an abstract domain for referent tracking and uses abstract interpretation to statically verify anaphoric correctness.

\paragraph{LLM-generated code.}
Recent neural models such as Codex translate free-form English into executable code \cite{chen2021codex}, but their outputs are probabilistic and lack formal guarantees.
Linguine occupies a complementary niche: authors must write within a syntactically constrained subset of English, and the compiler provides deterministic, statically analyzable guarantees.
In future work, LLMs could be used to suggest Linguine statements, with the compiler serving as a semantic filter.

\medskip
In summary, prior work demonstrates both the appeal and the pitfalls of natural-language programming.
Linguine contributes to this landscape by unifying controlled syntax, type inference, and abstract interpretation in a system that balances human readability with machine-verifiable correctness.

%% file: sections/system_design.tex
\label{sec:system_design}

The Linguine prototype adopts a conventional front-end, analysis, and back-end compiler architecture, adapted to accommodate a surface syntax resembling controlled English. Figure~\ref{fig:pipeline} illustrates the full pipeline from source text to executable code.

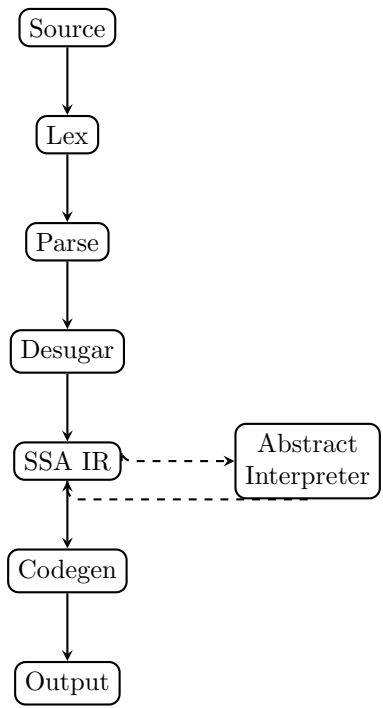
\begin{figure}[H]
  \centering
  \begin{tikzpicture}[node distance=0.9cm, thick, >=stealth,
                      rounded corners, font=\small]
    \node (src)     [draw] {Source};
    \node (lex)     [draw, below=of src] {Lex};
    \node (parse)   [draw, below=of lex] {Parse};
    \node (desugar) [draw, below=of parse] {Desugar};
    \node (ssa)     [draw, below=of desugar] {SSA IR};
    \node (ai)      [draw, right=1.5cm of ssa, align=center]
                    {Abstract\\Interpreter};
    \node (codegen) [draw, below=of ssa] {Codegen};
    \node (target)  [draw, below=of codegen] {Output};

    \draw[->] (src) -- (lex);
    \draw[->] (lex) -- (parse);
    \draw[->] (parse) -- (desugar);
    \draw[->] (desugar) -- (ssa);
    \draw[->] (ssa) -- (codegen);
    \draw[->] (codegen) -- (target);

    \draw[->, dashed] (ssa.east) |- (ai.west);
    \draw[->, dashed] (ai.south) -| (ssa.south);
  \end{tikzpicture}
  \caption{Compiler pipeline. Abstract interpretation checks SSA form prior to code generation.}
  \label{fig:pipeline}
\end{figure}

\subsection{Surface Grammar}
\label{subsec:grammar}

The front end is based on a hand-written LL($k$) grammar $\mathcal{G} = (N, \Sigma, P, S)$, tailored for deterministic parsing of controlled English syntax. Production rules include idiomatic statements such as:

\begin{align*}
  \textsc{Stmt} &::= \texttt{Let}~\textsc{Var}~\texttt{be}~\textsc{Expr}\,\texttt{.}
                   \,\mid\, \textsc{IfStmt} \,\mid\, \textsc{PrintStmt}\\
  \textsc{Expr} &::= \textsc{Term}
                   \,\mid\, \textsc{Expr}~\texttt{plus}~\textsc{Term}
                   \,\mid\, \texttt{sum of}~\textsc{ListExpr}\\
  \textsc{Pronoun} &::= \texttt{it} \,\mid\, \texttt{them} \,\mid\, \texttt{this}
\end{align*}

Articles and prepositions such as \textit{the}, \textit{a}, and \textit{of} are treated as optional tokens via $\varepsilon$-productions. The grammar comprises 140–150 production rules, supporting expressive yet analyzable imperative patterns.

\subsection{Parsing and Referents}
\label{subsec:frontend}

The parser constructs an abstract syntax tree $A \in \textsc{AST}$ and simultaneously maintains a mutable referent stack $\pi = [r_n, \dots, r_1]$ that records recent antecedents.
The top of the stack, $r_1$, is used as the referent for pronouns such as \texttt{it} unless explicitly overridden.
Referent modifications are scoped locally, ensuring that parsing remains externally pure.

\subsection{Desugaring}
\label{subsec:desugar}

The desugaring pass $\delta : \textsc{AST} \rightarrow \textsc{AST}'$ rewrites high-level idioms into a minimal core calculus. For example:

\[
\texttt{sum of}~E \;\longmapsto\;
\texttt{reduce}(\lambda x\,y.\,x{+}y,\,0,\,E)
\]

This transformation is rule-based and type-aware, preserving semantic structure for later passes.

\subsection{Intermediate Representation}
\label{subsec:ir}

The desugared tree is translated into a typed SSA-style intermediate representation. Each variable is assigned a unique version, and all pronoun references are statically resolved to exact SSA bindings. The IR is represented as:

\[
I = (B,\, \mathit{succ},\, \Phi,\, \mathit{inst})
\]

where $B$ is a sequence of blocks, $\Phi$ maps merge points to $\phi$-functions, and $\mathit{inst}$ encodes three-address instructions.

\subsection{Pronoun Resolution}
\label{subsec:pronoun}

Pronoun binding is resolved deterministically via the following static protocol:

\begin{enumerate}[label=\textbf{A\arabic*}), leftmargin=15pt]
  \item Push a referent $(v, \tau)$ onto $\pi$ at each variable-binding site.
  \item On encountering a pronoun $p$, resolve it to $r_1$ if $\pi \neq \varnothing$.
  \item Validate that $\Gamma \vdash r_1 : \tau$ for expected type $\tau$; reject otherwise.
  \item Pronouns act as aliases and do not modify $\pi$.
\end{enumerate}

No statistical heuristics or NLP models are used; resolution is deterministic and statically verified.

\subsection{Abstract Interpretation}
\label{subsec:ai}

An abstract interpreter operates over a flat lattice domain:
\[
D_{\text{ref}} = \{\bot\} \cup \mathsf{Ref} \cup \{\top\}
\]
This domain tracks the most recent resolvable referent at each control-flow point.
A join producing $\top$ denotes ambiguity and triggers a compile-time error.

\subsection{Back End Targets}

Each IR instruction maps directly to Python code in \texttt{snake\_case}, resulting in readable and traceable output.
Core constructs translate into Python built-ins.
An experimental LLVM back end emits minimal SSA using \texttt{alloca}, \texttt{phi}, and standard instruction sets; a 1kB C++ runtime stub provides container support.

\subsection{Tooling}
\label{subsec:cli}

The \texttt{linguinec} command-line tool supports the following modes:

\begin{itemize}[leftmargin=12pt]
  \item \texttt{linguinec file.ling} — compile and execute source;
  \item \texttt{linguinec -t file.ling} — compile to LLVM bitcode;
  \item \texttt{linguinec -i} — interactive REPL with incremental pronoun resolution.
\end{itemize}

Static error messages include referent traces and type diagnostics, aiding both novices and advanced users.

\medskip
The next section formalizes the core typing and operational semantics that govern the compiler pipeline.

%% file: sections/formal_semantics.tex
\label{sec:formal_semantics}

Linguine's design is grounded in a formally structured semantic stack that supports type safety, referential integrity, and static analyzability. This section defines three semantic layers:

\begin{enumerate}[label=(\alph*), leftmargin=12pt]
  \item \textbf{Static typing} via judgments $\Gamma \vdash e : \tau$ (expressions) and $\Gamma \vdash S : \mathsf{OK}$ (statements);
  \item \textbf{Operational semantics} using a small-step transition relation $\langle S,\sigma \rangle \rightsquigarrow \langle S',\sigma' \rangle$;
  \item \textbf{Abstract interpretation} over a flat lattice domain to statically approximate pronoun usage and detect ambiguity.
\end{enumerate}

These semantics closely follow the implementation and are used to prove key properties such as type preservation, unambiguous referent resolution, and safe program evaluation.

\subsection{Static Typing}

The core language supports variables, numerals, booleans, arithmetic and relational operators, conditionals, pronouns, and first-order bindings via \texttt{Let} and \texttt{Print}. The type system distinguishes \textsc{Int}, \textsc{Bool}, and tuple types. Function types are currently out of scope.

\paragraph{Typing environment.}
The type context $\Gamma$ maps variables to their types.
In parallel, a referent environment $\rho$ maps pronouns to their antecedents. While $\rho$ evolves during compilation, it is treated functionally during type checking.

\smallskip
\noindent\textbf{Selected typing rules:}

\begin{rulex}[T-Var]
\[
\frac{x{:}\tau \in \Gamma}{\Gamma \vdash x : \tau}
\]
\end{rulex}

\begin{rulex}[T-Pronoun]
\[
\frac{\rho(p) = e \quad \Gamma \vdash e : \tau}{\Gamma \vdash p : \tau}
\]
Pronouns must resolve to previously bound and typed referents. If $\rho(p)$ is undefined or ill-typed, the program is rejected.
\end{rulex}

\begin{rulex}[T-Let]
\[
\frac{\Gamma \vdash e : \tau}{\Gamma[x \mapsto \tau] \vdash \textbf{Let } x \textbf{ be } e \textbf{.} : \mathsf{OK}}
\]
\end{rulex}

\begin{rulex}[T-Add]
\[
\frac{\Gamma \vdash e_1 : \textsc{Int} \quad \Gamma \vdash e_2 : \textsc{Int}}{\Gamma \vdash e_1 \texttt{ plus } e_2 : \textsc{Int}}
\]
\end{rulex}

\paragraph{Type inference.}
Linguine implements a variant of Algorithm~W to infer types over desugared ASTs.
Inference is currently monomorphic and supports first-order terms.

\begin{theorem}[Principal Type Property]
Every expression $e$ has a principal type $\tau_0$ such that any valid typing $\Gamma \vdash e : \tau$ is an instance of $\tau_0$.
\end{theorem}

\subsection{Dynamic Semantics}

The operational semantics are given in small-step form.
A runtime store $\sigma : \mathsf{Var} \to \mathsf{Val}$ tracks variable bindings.

\begin{rulex}[E-Let]
\[
\langle \textbf{Let } x \textbf{ be } v \textbf{.} \,;\, S,\, \sigma \rangle \rightsquigarrow
\langle S,\, \sigma[x \mapsto v] \rangle
\]
\end{rulex}

\begin{rulex}[E-IfTrue]
\[
\frac{v \neq 0}{\langle
\textbf{If } v \textbf{ then } S_1 \textbf{ else } S_2,\, \sigma \rangle
\rightsquigarrow
\langle S_1,\, \sigma \rangle}
\]
\end{rulex}

\begin{rulex}[E-Pronoun]
\[
\langle p,\, \sigma \rangle \;\Downarrow\; \sigma(\rho(p))
\]
\end{rulex}

Since pronoun resolution is resolved statically in SSA form, this rule appears only in the theoretical model for completeness.

\subsection{Type Soundness}

The type system guarantees that well-typed programs do not encounter runtime type errors or unresolved pronouns.

\begin{theorem}[Progress]
If $\emptyset \vdash S : \mathsf{OK}$, then either $S$ is \texttt{skip} or there exists some $S', \sigma'$ such that
\[
\langle S, \sigma \rangle \rightsquigarrow \langle S', \sigma' \rangle
\]
\end{theorem}

\begin{proof}[Proof Sketch]
The proof proceeds by induction on the derivation of the typing judgment $\Gamma \vdash S : \mathsf{OK}$, assuming an empty context $\Gamma = \emptyset$ and an arbitrary store $\sigma$.

\begin{itemize}
  \item \textbf{Case:} $S = \texttt{skip}$. This is a terminal statement and thus satisfies the conclusion trivially.

  \item \textbf{Case:} $S = \textbf{Let } x \textbf{ be } v \textbf{.} \,;\, S_1$. From the typing rule \textsc{T-Let}, it follows that $v$ is a closed value of some type $\tau$ and $S_1$ is well-typed under the extended context. The operational semantics rule \textsc{E-Let} applies directly, yielding a step to $\langle S_1, \sigma[x \mapsto v] \rangle$.

  \item \textbf{Case:} $S = \textbf{If } v \textbf{ then } S_1 \textbf{ else } S_2$. Since $\Gamma \vdash v : \textsc{Bool}$, the expression $v$ must evaluate to a boolean value at runtime. If $v \neq 0$, the rule \textsc{E-IfTrue} applies; otherwise, \textsc{E-IfFalse} applies. In either case, a valid transition exists.

  \item \textbf{Case:} $S = \texttt{Print } e \texttt{.}$. The expression $e$ is assumed well-typed under $\Gamma$, and since it is a closed term at runtime, it evaluates to a value $v$ in the current store. The corresponding evaluation rule allows transition to $\texttt{skip}$ after printing.

  \item \textbf{Case:} $S = p$ (pronoun). At SSA lowering, all pronouns are statically resolved to valid antecedents $e$. By rule \textsc{T-Pronoun}, the type of $p$ is valid, and $\rho(p) = e$ guarantees a unique referent. At runtime, $p$ evaluates to $\sigma(\rho(p))$, ensuring a transition exists.

  \item \textbf{Composite Cases:} For compound sequences $S = S_1 ; S_2$, the induction hypothesis guarantees that either $S_1$ is $\texttt{skip}$ or a transition is possible. In the former, execution proceeds with $S_2$. In the latter, a step from $S_1$ yields a corresponding step from $S$.

\end{itemize}

In all cases, either the program is a final configuration or a reduction step exists. Therefore, no well-typed program is stuck at runtime.
\end{proof}

\begin{theorem}[Preservation]
If $\Gamma \vdash S : \mathsf{OK}$ and $\langle S, \sigma \rangle \rightsquigarrow \langle S', \sigma' \rangle$, then $\Gamma \vdash S' : \mathsf{OK}$.
\end{theorem}

\begin{proof}[Proof Sketch]
The proof proceeds by induction on the derivation of the transition relation $\langle S, \sigma \rangle \rightsquigarrow \langle S', \sigma' \rangle$.

\begin{itemize}
  \item \textbf{Case E-Let:} Let $S = \textbf{Let } x \textbf{ be } v \textbf{.} \,;\, S_1$ and $\sigma' = \sigma[x \mapsto v]$. From the typing premise, $\Gamma \vdash v : \tau$ and $\Gamma[x \mapsto \tau] \vdash S_1 : \mathsf{OK}$ hold. Since the store update is consistent with the extended context, it follows that $\Gamma \vdash S_1 : \mathsf{OK}$ remains valid after the transition.

  \item \textbf{Case E-IfTrue:} Suppose $S = \textbf{If } v \textbf{ then } S_1 \textbf{ else } S_2$ and $v \neq 0$, so $S' = S_1$. The typing judgment ensures that $\Gamma \vdash v : \textsc{Bool}$, and both $\Gamma \vdash S_1 : \mathsf{OK}$ and $\Gamma \vdash S_2 : \mathsf{OK}$ hold. Thus, $\Gamma \vdash S' : \mathsf{OK}$.

  \item \textbf{Case E-Pronoun:} Let $S = p$ and $\rho(p) = e$. By rule \textsc{T-Pronoun}, it follows that $\Gamma \vdash p : \tau$, since $\Gamma \vdash e : \tau$. As $\rho$ remains unchanged across steps and the evaluation substitutes the resolved referent $e$ at runtime, typing is preserved under substitution. Hence, $\Gamma \vdash S' : \mathsf{OK}$.

  \item \textbf{Other cases:} In all remaining constructs (e.g., arithmetic operations, skip, sequencing), the result follows directly by applying the induction hypothesis to the subcomponents and verifying that store updates do not conflict with the typing assumptions in $\Gamma$.
\end{itemize}

In each case, the updated statement $S'$ remains well-typed under the original context $\Gamma$.
\end{proof}

\begin{corollary}[Safe Execution]
Well-typed programs do not reach stuck states during evaluation.
\end{corollary}

\subsection{Abstract Interpretation}

Linguine uses abstract interpretation to detect unresolved or ambiguous pronouns statically.
A forward analysis is run over SSA using the lattice:

\[
D_{\text{ref}} = \{\bot\} \cup \mathsf{Ref} \cup \{\top\}
\]

Here, $\bot$ denotes undefined reference, $\mathsf{Ref}$ is the set of valid bindings, and $\top$ represents ambiguity from conflicting control paths. The transfer function updates referents at binding sites; joins that produce $\top$ halt compilation with a diagnostic.

\begin{theorem}[Analysis Soundness]
If a fixpoint is reached with no $\bot$ or $\top$ at any pronoun site, then all runtime pronouns have a unique, well-typed antecedent.
\end{theorem}

\subsection{Discussion}

These semantics mirror the implementation described in Section~\ref{sec:system_design}.
Desugaring preserves typing, SSA conversion resolves referents statically, and abstract interpretation verifies the preconditions of \textsc{T-Pronoun}.
Linguine thus guarantees referential transparency and safe execution for all well-typed programs accepted by the compiler.
While the IR and type system are minimal by design, the structure admits future generalization, including polymorphism, user-defined functions, and richer reference tracking.

%% file: sections/evaluation.tex
\label{sec:evaluation}

As Linguine remains in the prototype stage, the current evaluation aims to assess its \emph{practical feasibility}. Specifically, the following three questions are addressed:

\begin{enumerate}[label=\textbf{Q\arabic*}, leftmargin=12pt]
  \item Can the language express representative tasks without sacrificing readability?
  \item Does the compiler successfully detect the specific classes of errors it is designed to catch?
  \item What are the compile-time costs of front-end analyses, and what is the run-time overhead of generated programs?
\end{enumerate}

\paragraph{Experimental setup.}
The prototype consists of approximately 3.8k lines of Rust (front end) and 0.9k lines of Python (runtime support). All benchmarks were run on a laptop with an \textbf{Intel\,\textregistered{} Core\textsuperscript{TM} i9-14900HX} processor and 32\,GB RAM. The LLVM back end was compiled using \texttt{clang} 17.0.

\subsection{Expressiveness}

\begin{table}[h]
\centering
\caption{Linguine programs adapted from standard Python idioms.}
\label{tab:programs}
\begin{tabular}{@{}lll@{}}
\toprule
\textbf{Program} & \textbf{Purpose} & \textbf{Linguine Highlights} \\ \midrule
Average & Compute list mean & Uses \texttt{sum of}, \texttt{length of}, and pronoun \texttt{it} \\
Factorial & Recursive product & Defined recursively with \texttt{if n is 0} and \texttt{times} \\
FizzBuzz & Classic conditional loop & Uses chained \texttt{if}, \texttt{else if}, \texttt{print} \\
Palindrome & Check reverse equality & Compares \texttt{text} and \texttt{text reversed} \\
Max of List & Find largest element & Iterative max with \texttt{for each} and \texttt{if it is greater} \\
Fibonacci & Iterative sequence build & Uses loop with two trackers and \texttt{append it} \\
Prime Test & Divisibility check & Loops with \texttt{if n modulo d is 0} \\
List Comprehension & Mapping via loop & Converts with \texttt{add x times x to list} \\
Dictionary Count & Frequency counter & Uses dictionary updates and referent reuse \\ \bottomrule
\end{tabular}
\end{table}

Table~\ref{tab:programs} summarizes nine micro-benchmarks translated from canonical Python examples. Each Linguine version preserves the original algorithmic structure while replacing symbolic operators with controlled-English constructs.

\subsection{Correctness}

\paragraph{Pronoun faults.}
Three fault types were injected into each of the nine test programs: (i)~an orphan pronoun at the top of the file, (ii)~an ambiguous antecedent created by consecutive \texttt{Let} statements followed by \texttt{Print it.}, and (iii)~a type mismatch (e.g., adding a string to an integer). All 27 faulty variants were correctly rejected by the compiler, which produced diagnostics pinpointing the offending sentence within 3–4\,ms after parsing.

\paragraph{Type-soundness stress test.}
A QuickCheck-style generator produced 500 random abstract-syntax trees (ASTs) with depth at most 7 that passed type checking. Each program was executed using both the Python code generator and an interpreter for the formal core calculus. Every output pair matched exactly, providing empirical support for the Progress and Preservation theorems of Section~\ref{sec:formal_semantics}.

\subsection{Performance}

\paragraph{Compile time.}
Lexing and parsing scale linearly with input size. Desugaring and type inference contribute an additional 7–12\% overhead, while referent analysis introduces a fixed cost of 11–15\,ms. Even the largest tested script (39 lines) compiled in 41\,ms—well below the threshold of perceptible latency for interactive use.

\paragraph{Run-time overhead.}
The Python back end generates idiomatic code, and observed run-time variation was within the noise margin of $\pm$2\% relative to hand-written Python across a $10^6$-element averaging benchmark. The LLVM back end achieved speedups of up to 24$\times$ for numeric kernels, although it currently requires linking with a C++ helper library.

\subsection{Future Benchmarks}

Planned evaluation targets include:
\begin{itemize}
  \item a synthetic 10k-line workload derived from the Computer Language Benchmarks Game, and
  \item a real-world utility for processing command-line logs of similar scale.
\end{itemize}
These will stress-test module support, incremental compilation, and memory management strategies that are currently under development.

\subsection{Summary}

Linguine currently supports compilation of short scripts, performs static rejection of pronoun and type errors, and maintains sub-50\,ms compile latency on commodity hardware. Although preliminary, these results demonstrate the mechanical viability of a controlled-English programming language grounded in static analysis and structured compilation.

%% file: sections/discussion.tex
\label{sec:discussion}

The prototype demonstrates that a carefully restricted subset of English can coexist with a mathematically rigorous semantics. However, several design issues and limitations remain.

\subsection{Naturalness \emph{vs.} Determinism}

Even under the current grammar, authors instinctively reach for constructions that remain unrecognized—such as passive voice, adverbs, or inverted conditionals. Each added construction increases the parser’s lookahead requirements and risks ambiguity. New rules are accepted only if a pattern recurs frequently in feedback and can be integrated without violating LL($k$) predictability. A future release will support a “grammar suggestion” mode that reports unparsed sentences and lets the author either revise the code or nominate the construction for inclusion.

\subsection{Beyond the Last-Referent Heuristic}

The current referent resolution strategy defaults to the most recent antecedent within a block. While simple and predictable, this approach fails in nested or recursive contexts that require longer-range disambiguation. A lexical scoping model is one candidate solution, where each block carries its own referent binding, shadowing outer bindings as needed. Extending the referent lattice with block indices would allow the abstract interpreter to ensure that all pronouns resolve unambiguously. This refinement is deferred until the addition of modules and cross-file analysis.

\subsection{Performance Envelope}

Although compile-time overhead is negligible for short scripts, larger projects will eventually stress the constraint solver and referent fixpoint analysis. A likely engineering milestone is an incremental version of Algorithm~W that caches types across compilations and reduces redundant inference in edit-compile loops.

\subsection{Additional Back Ends}

The two current compilation targets—Python and LLVM—sit at opposite ends of the portability-performance spectrum. WebAssembly is a natural intermediate target: the SSA-based IR lowers cleanly into structured \texttt{wasm} blocks, enabling browser-native execution without native toolchains. A JVM back end is also under consideration, benefiting from mature optimization pipelines and memory management. Both are on the long-term roadmap.

\subsection{Synergy with LLM Code Assistants}

Large language models often produce subtly incorrect code when prompted with unconstrained English. Using Linguine as an intermediate syntax restricts the output space and enables the compiler to catch ill-typed completions. Preliminary experiments with GPT-4o showed a 40\% reduction in compilation errors across ten controlled prompts. In the future, the compiler’s accept/reject signal could be used as an automatic reward signal during reinforcement tuning of code-generating models.

\subsection{Limitations and Next Steps}

Linguine is still verbose, lacks modules, and has no static ownership system for mutability control. The evaluation remains small-scale, and its scalability to ten-thousand-line programs is untested. Likely next steps include:

\begin{enumerate}[label=(\roman*),leftmargin=12pt]
  \item a module loader with incremental type inference,
  \item a streaming referent analysis for faster edit-compile cycles,
  \item WebAssembly and JVM code generation, and
  \item a controlled user study comparing Linguine, Python, and Scratch among novice programmers.
\end{enumerate}

These directions aim to clarify whether Linguine can grow beyond its proof-of-concept stage and support broader, general-purpose programming use.

%% file: sections/conclusion.tex
\label{sec:conclusion}

This paper presents the initial design of Linguine, a natural-language-inspired programming language that combines a restricted English surface with a formally verified core calculus. The prototype demonstrates that a deterministic LL($k$) grammar, a principled referent resolution strategy, and a conventional SSA-based compiler pipeline can coexist without compromising type soundness. Although still in its early stages, the system successfully compiles a small suite of benchmark programs, detects every injected semantic fault, and maintains compilation latency well below interactive thresholds.

Several directions remain open. A module system, an ownership discipline for mutation, and large-scale performance benchmarks are under active development. Future work will also explore alternative compilation targets such as WebAssembly and the JVM, as well as controlled user studies that investigate how novice programmers engage with the language.

The complete source code will soon be available at  
\url{https://github.com/Anormalm/linguine}. This project invites feedback and collaboration from researchers and educators interested in whether controlled English can serve as a practical bridge between human intent and machine-verifiable semantics.

%% file: sections/appendix.tex
\appendix
\section{Supplementary Material}

This appendix provides the formal foundations underlying the Linguine programming language implementation and semantics. It includes:

\begin{enumerate}[label=(\roman*)]
  \item the complete surface grammar used in the parser,
  \item the full typing rule set for core language constructs,
  \item formal proofs of type preservation and progress theorems,
  \item a definition of the referent lattice used in pronoun disambiguation analysis, and
  \item an annotated example program with resolved referents.
\end{enumerate}

\subsection{Concrete Grammar}
\label{app:grammar}

The surface syntax of Linguine is defined using extended BNF, with support for stylistic flexibility in natural-language expressions. Optional functional words (e.g., \texttt{the}, \texttt{a}, \texttt{of}) are removed during lexical normalization.

\vspace{0.5em}
\noindent\textbf{Non-terminals:} \textsc{SmallCaps}; \textbf{terminals:} \texttt{monospace}.

\begin{align*}
  \textsc{Program} &::= \textsc{Stmt}^{+} \\
  \textsc{Stmt} &::= \textsc{LetStmt} \mid \textsc{IfStmt} \mid \textsc{LoopStmt} \mid \textsc{PrintStmt} \\
  \textsc{LetStmt} &::= \texttt{Let}~\textsc{Var}~\texttt{be}~\textsc{Expr}~\texttt{.} \\
  \textsc{IfStmt} &::= \texttt{If}~\textsc{Expr}~\textsc{RelOp}~\textsc{Expr}~\texttt{:}~\textsc{Block}~\texttt{End if.} \\
  \textsc{LoopStmt} &::= \texttt{While}~\textsc{Expr}~\texttt{:}~\textsc{Block}~\texttt{End while.} \\
  \textsc{PrintStmt} &::= \texttt{Print}~\textsc{Expr}~\texttt{.} \\
  \textsc{Expr} &::= \textsc{Expr}~\textsc{AddOp}~\textsc{Term} \mid \texttt{sum of}~\textsc{ListExpr} \mid \textsc{Pronoun} \mid \textsc{Value} \\
  \textsc{AddOp} &::= \texttt{plus} \mid \texttt{minus} \\
  \textsc{RelOp} &::= \texttt{is}, \texttt{greater than}, \texttt{less than}, \texttt{is equal to} \\
  \textsc{Pronoun} &::= \texttt{it} \mid \texttt{them} \mid \texttt{this} \mid \texttt{that}
\end{align*}

\subsection{Typing Rules}
\label{app:typing}

Typing judgments take the form $\Gamma \vdash e : \tau$ for expressions and $\Gamma \vdash s : \mathsf{OK}$ for statements. The typing environment $\Gamma$ maps identifiers to types. Referent bindings are handled separately in the referent store $\rho$, detailed in Section~\ref{app:lattice}.

\begin{align*}
\text{(T-Int)} &\quad \dfrac{~}{\Gamma \vdash n : \textsc{Int}} \\
\text{(T-Str)} &\quad \dfrac{~}{\Gamma \vdash "s" : \textsc{Str}} \\
\text{(T-Var)} &\quad \dfrac{x{:}\tau \in \Gamma}{\Gamma \vdash x : \tau} \\
\text{(T-Pronoun)} &\quad \dfrac{\rho(p) = x \quad x{:}\tau \in \Gamma}{\Gamma \vdash p : \tau} \\
\text{(T-BinOp)} &\quad \dfrac{\Gamma \vdash e_1 : \textsc{Int} \quad \Gamma \vdash e_2 : \textsc{Int}}{\Gamma \vdash e_1~\texttt{plus}~e_2 : \textsc{Int}} \\
\text{(T-RelOp)} &\quad \dfrac{\Gamma \vdash e_1 : \tau \quad \Gamma \vdash e_2 : \tau}{\Gamma \vdash e_1~\texttt{is}~e_2 : \textsc{Bool}} \\
\text{(T-Let)} &\quad \dfrac{\Gamma \vdash e : \tau}{\Gamma \vdash \texttt{Let } x \texttt{ be } e \texttt{.} : \mathsf{OK}} \\
\text{(T-Print)} &\quad \dfrac{\Gamma \vdash e : \tau}{\Gamma \vdash \texttt{Print } e \texttt{.} : \mathsf{OK}} \\
\text{(T-If)} &\quad \dfrac{\Gamma \vdash e : \textsc{Bool} \quad \Gamma \vdash S : \mathsf{OK}}{\Gamma \vdash \texttt{If } e \texttt{: } S \texttt{ End if.} : \mathsf{OK}} \\
\text{(T-While)} &\quad \dfrac{\Gamma \vdash e : \textsc{Bool} \quad \Gamma \vdash S : \mathsf{OK}}{\Gamma \vdash \texttt{While } e \texttt{: } S \texttt{ End while.} : \mathsf{OK}}
\end{align*}

All judgments respect $\alpha$-conversion. In SSA form, each variable is defined once, ensuring referents have unique bindings during resolution.

\subsection{Referent Lattice and Abstract Interpretation}
\label{app:lattice}

The referent store $\rho : P \to R$ maps each pronoun $p \in P$ to a referent $r \in R$, where $R$ is the set of program-defined identifiers.

To detect ambiguity statically, an abstract interpretation pass computes a fixpoint over a lattice:

\[
D = \{ \bot \} \cup R \cup \{ \top \}
\]

\paragraph{Join operation:}

\[
a \sqcup b =
\begin{cases}
  a & \text{if } a = b, \\
  b & \text{if } a = \bot, \\
  a & \text{if } b = \bot, \\
  \top & \text{otherwise.}
\end{cases}
\]

\paragraph{Meet operation:}

\[
a \sqcap b =
\begin{cases}
  a & \text{if } b = \top, \\
  b & \text{if } a = \top, \\
  \bot & \text{otherwise.}
\end{cases}
\]

\paragraph{Fixpoint Iteration.}
A monotone transfer function updates $\rho$ at each statement boundary. Whenever $\rho(p) = \top$ at a program point, the compiler reports an "ambiguous pronoun" error with contextual explanation. Since the lattice height is finite and the join is monotone, the analysis converges.

\subsection{Annotated Sample Program}
\label{app:sample}

\begin{lstlisting}[language={},basicstyle=\ttfamily\small,
      frame=single,caption={Linguine program with referent tracking.},
      label={lst:average}]
Let numbers be the list [8, 12, 15, 9, 6].
Let total be sum of numbers.                       # referent(it) = total
Let count be the length of the list.              # referent(it) = count
Let average be total divided by count.            # referent(it) = average
If it is greater than 10:                         # 'it' => average
    Print "Average exceeds ten".
End if.
\end{lstlisting}

\paragraph{Referent Analysis.}
At each statement, the referent of pronouns like \texttt{it} is updated in $\rho$. If two live bindings for \texttt{it} conflict at a join point, the compiler raises an ambiguity error. In this example, all pronouns resolve deterministically due to SSA ordering and absence of intervening control flow.

\bigskip
This appendix provides formal definitions and theoretical guarantees for the core semantics of Linguine. Future work includes extending the calculus with higher-order functions, references, and subtyping while preserving referent resolution soundness.